\documentclass[aps,pra,twocolumn,superscriptaddress,longbibliography,floatfix]{revtex4-2}

\usepackage{amsmath,amssymb,amsthm}
\newtheorem{theorem}{Theorem}[section]
\newtheorem{lemma}[theorem]{Lemma}
\newtheorem{proposition}[theorem]{Proposition}
\usepackage{graphicx}
\usepackage{xcolor}
\usepackage{booktabs}
\usepackage{hyperref}
\usepackage{bm}
\usepackage{array}
\usepackage{tikz}
\usetikzlibrary{shapes.geometric,positioning,calc,arrows.meta}
\usepackage[caption=false]{subfig}

\hypersetup{colorlinks=true, linkcolor=blue, citecolor=blue, urlcolor=blue}

\newcommand{\vt}{\vartheta}
\newcommand{\al}{\alpha}
\newcommand{\De}{\Delta}
\newcommand{\Tr}{\operatorname{Tr}}

\newcommand{\QC}{\ensuremath{\text{Quad-}C_5}}

\tikzset{
  hub/.style  = {circle, draw=black, fill=black!75, text=white,
                 font=\footnotesize\bfseries, inner sep=1.8pt, minimum size=18pt},
  leaf/.style = {circle, draw=black, fill=black!30, text=white,
                 font=\footnotesize\bfseries, inner sep=1.8pt, minimum size=18pt},
  gedge/.style= {thick},
  diagedge/.style = {thick, dashed},
}

\begin{document}

\title{The Quad-$C_5$ Graph: Maximum Contextuality Gap on Eight Vertices}

\author{U\u{g}ur Tamer}
\email{utamer24@ku.edu.tr}
\affiliation{Department of Physics, Ko\c{c} University, 34450 Sar{\i}yer, Istanbul, T\"{u}rkiye}

\author{Berke \"{U}nal}
\email{bunal21@ku.edu.tr}
\affiliation{Department of Mathematics, Ko\c{c} University, 34450 Sar{\i}yer, Istanbul, T\"{u}rkiye}

\author{Alper Dizdar}
\email{dizdar@istanbul.edu.tr}
\affiliation{Department of Physics, Faculty of Science,
University of Istanbul, 34134 Vezneciler, Istanbul, T\"{u}rkiye}

\author{Zafer Gedik}
\email{zafer.gedik@sabanciuniv.edu}
\affiliation{Faculty of Engineering and Natural Sciences,
Sabanc{\i} University, Tuzla, Istanbul 34956, T\"{u}rkiye}

\author{\"{O}zg\"{u}r E.\ M\"{u}stecapl\i{}o\u{g}lu}
\email{omustecap@ku.edu.tr}
\affiliation{Department of Physics, Ko\c{c} University, 34450 Sar{\i}yer, Istanbul, T\"{u}rkiye}
\affiliation{T\"{U}B\.{I}TAK Research Institute for Fundamental Sciences (TBAE), 41470 Gebze, T\"{u}rkiye}

\date{\today}

\begin{abstract}
Quantum measurements can exhibit contextuality: their outcomes cannot always
be explained by assigning pre-existing values that are independent of which
compatible measurements are performed together. The
Klyachko--Can--Binicio\u{g}lu--Shumovsky (KCBS) inequality provides the
canonical minimal test of this effect for a single three-level quantum system,
or qutrit, using five measurement events arranged as a pentagon. Here we ask
whether a larger but still compact set of measurement events can produce a
stronger separation between quantum predictions and the corresponding
noncontextual limit. We perform an exhaustive search over all 11{,}117
connected non-isomorphic graphs with eight vertices, where each vertex
represents a measurement event and edges connect pairs of events that cannot
occur together. We identify a sparse ten-edge graph, which we call \QC, as the
unique maximizer of this separation at the reported numerical precision. The
graph can be understood as four overlapping KCBS pentagons, with every edge
shared by two pentagons. \QC\ already demonstrates contextuality in a qutrit,
for which we obtain an exact analytical result and find the same violation
above the noncontextual bound as in the original KCBS test. The larger
quantum--noncontextual separation allowed by the graph, however, becomes
accessible in a four-level quantum system, where numerical optimization
reaches the full graph-theoretic quantum bound and yields a larger
contextuality gap than the standard eight-vertex Wagner-graph benchmark while
requiring fewer pairwise constraints. \QC\ therefore provides a compact
connection between minimal qutrit contextuality and stronger contextuality
tests available in higher-dimensional quantum systems.
\end{abstract}

\maketitle

\section{Introduction}

Quantum contextuality~\cite{KochenSpecker1967} expresses the impossibility, in
general, of explaining quantum measurement outcomes by assigning pre-existing
values that are independent of which compatible measurements are performed
together. Unlike Bell nonlocality, contextuality does not require spatially
separated subsystems and can already appear in a single finite-dimensional
quantum system. This makes it both a foundational signature of nonclassicality
and a resource for quantum information processing, with established
connections to quantum computation, communication, and
sensing~\cite{Budroni2022,Howard2014,Raussendorf2013}.

Identifying simple yet strong contextuality witnesses is an important problem
in quantum contextuality. The Klyachko--Can--Binicio\u{g}lu--Shumovsky (KCBS)
construction demonstrated that five cyclically compatible measurements on a
single spin-1 system are sufficient for a state-dependent contextuality
test~\cite{Klyachko2002,Klyachko2008}. This pentagonal structure became the
canonical minimal qutrit example and was subsequently embedded in the broader
family of $n$-cycle contextuality scenarios, for which noncontextual
inequalities and their quantum violations admit analytic
characterization~\cite{Araujo2013}. Such closed-form results are rare because
separating classical, quantum, and more general probabilistic correlations is,
in general, a difficult problem.

Most experimental contextuality tests have focused on minimal or highly
structured scenarios, such as the KCBS pentagon for qutrits and $n$-cycle
inequalities~\cite{Araujo2013}. These examples provide fundamental tests of
contextuality, while recent resource-oriented developments motivate a broader
question: how can contextuality be organized and strengthened in larger but
still experimentally manageable measurement scenarios? This question is
relevant because contextuality has been connected to advantages in quantum
computation~\cite{Howard2014,Raussendorf2013},
communication~\cite{Gupta2023}, and metrology~\cite{Lostaglio2020}. Here we
address the eight-vertex optimization problem beyond the analytically
understood cycle family and the standard Wagner
benchmark~\cite{Cabello2014,Lovasz1979}.

A systematic approach to this question is provided by the graph-theoretic
framework of Cabello, Severini, and Winter (CSW)~\cite{Cabello2014}. Each
measurement event is represented as a vertex, and mutually exclusive events,
meaning events that cannot occur together, are connected by edges, yielding an
exclusivity graph $G=(V,E)$. For projective measurements, the total event
probability satisfies
\begin{equation}
  \al(G) \leq \sum_{i\in V}p_i \leq \vt(G),
  \label{eq:csw_intro}
\end{equation}
where $\al(G)$ is the independence number, giving the noncontextual bound, and
$\vt(G)$ is the Lov\'{a}sz theta number~\cite{Lovasz1979}, giving the projective
quantum bound~\cite{Cabello2014}. The contextuality gap
\begin{equation}
  \De(G) = \vt(G) - \al(G)
  \label{eq:gap_intro}
\end{equation}
therefore measures the difference between the quantum and noncontextual bounds
within this graph model. A larger gap corresponds to a larger separation
between these two descriptions.

Maximizing the contextuality gap for a fixed vertex count is a nontrivial
combinatorial problem. Adding edges introduces additional exclusivity
constraints that can suppress the quantum value, yet may also lower the
noncontextual bound and reorganize the geometry of the witness. Consequently,
the densest or most symmetric graph need not yield the largest contextuality
gap, and even the eight-vertex case presents a non-obvious optimization
landscape.

Known examples illustrate this hierarchy. The pentagon $C_5$ realizes the KCBS
witness, and the seven-cycle $C_7$ belongs to the analytically characterized
$n$-cycle family~\cite{Araujo2013}. At eight vertices, the Wagner graph has
served as a standard non-cycle benchmark~\cite{Cabello2014,Lovasz1979}. The
full eight-vertex landscape, however, extends well beyond these symmetric
examples, and it is not a priori clear whether any known graph achieves the
global maximum gap or whether a different structure performs better.

A closely related graph-enumeration study considered fully contextual
correlations by computing the independence number, Lov\'{a}sz theta number, and
fractional packing number for non-isomorphic graphs with fewer than $11$
vertices~\cite{Amselem2012}. The fractional packing number is a
linear-programming relaxation of the independence number in which vertices may
carry fractional weights, constrained so that the total weight assigned to a
clique of mutually exclusive events does not exceed one. The resulting graph
data were later used to identify experimentally accessible graphs maximizing
the ratio $\vt(G)/\al(G)$, a relative measure of contextuality
strength~\cite{Canas2016}. This ratio has since been studied as a graph-based
dimension witness~\cite{Ray2021}, and the CSW exclusivity-graph framework has
also been extended to linear games and Bell
inequalities~\cite{Rosicka2016}. Here we consider a complementary
optimization criterion, the absolute contextuality gap
$\De(G)=\vt(G)-\al(G)$, which directly measures the difference between the
projective quantum and noncontextual bounds. Maximizing $\vt(G)/\al(G)$ favors
large relative violations, whereas maximizing $\vt(G)-\al(G)$ favors a large
absolute difference between the two bounds. These two criteria need not select
the same graph, and for eight vertices they do not. We use the latter criterion
to identify the optimal connected eight-vertex graph and then analyze its
structure, dimension dependence, and noise robustness.

We resolve this problem by an exhaustive semidefinite-programming search over
all connected eight-vertex graphs. 

At the reported numerical precision, the global maximum is attained by a ten-edge graph, which we call \QC, that outperforms both the
Wagner graph and every cycle in this vertex range. Its structure can be
understood in terms of four overlapping KCBS pentagons whose edges are shared
in a regular pattern. In fact, every edge of \QC\ belongs to exactly two of the
four pentagons, providing a compact way of combining KCBS-type contextuality
within a larger graph.

We additionally determine how the contextuality of \QC\ depends on
Hilbert-space dimension. The eight graph vertices label eight measurement
events; they do not imply an eight-dimensional quantum system. An explicit
three-dimensional orthogonal representation shows that the graph is already
contextual for a qutrit, i.e., a three-level quantum system, and we determine
this qutrit value analytically. Numerical evidence shows that the larger
quantum--noncontextual separation allowed by the graph is reached in a
four-dimensional Hilbert space, corresponding to a four-level quantum system.
Thus, \QC\ connects minimal qutrit contextuality with a stronger
four-dimensional contextuality witness; we also compare its noise robustness
with those of KCBS and the Wagner graph.

\section{Mathematical Background}
\label{sec:background}

\subsection{Lov\'{a}sz Theta and Semidefinite Programming}
\label{sec:theta_sdp}

For a graph $G$ on $n$ vertices with edge set $E$, the Lov\'{a}sz theta
function can be written as the semidefinite program~\cite{Lovasz1979,Knuth1994}
\begin{equation}
  \vt(G) = \max_{\substack{X\succeq 0,\;\Tr(X)=1\\
            X_{ij}=0\;\forall(i,j)\in E}}\;
            \mathbf{1}^\top X\,\mathbf{1},
  \label{eq:sdp}
\end{equation}
where $X$ is an $n\times n$ real symmetric matrix and
$\mathbf{1}=(1,\ldots,1)^\top$ is the all-ones column vector. The notation
$X\succeq0$ means that $X$ is positive semidefinite, i.e.,
$\mathbf{y}^\top X\mathbf{y}\geq0$ for every real vector $\mathbf{y}$.
Thus Eq.~\eqref{eq:sdp} maximizes a linear function of $X$ subject to the
normalization $\Tr(X)=1$, the edge constraints $X_{ij}=0$, and the
positive-semidefinite constraint. The symbol $\succeq$ here denotes the
standard semidefinite ordering.

An equivalent formulation uses an orthonormal representation of the
graph~\cite{Lovasz1979,Knuth1994}:
\begin{equation}
  \vt(G) = \max_{|\psi\rangle,\{|v_i\rangle\}}
           \sum_{i=1}^n |\langle\psi|v_i\rangle|^2,
  \label{eq:orth_rep}
\end{equation}
where $\{|v_i\rangle\}\subset\mathbb{C}^d$ are unit vectors satisfying
$\langle v_i|v_j\rangle=0$ whenever $(i,j)\in E$. The normalized vector
$|\psi\rangle$ is conventionally called the \emph{handle} of the
orthonormal representation~\cite{Lovasz1979,Knuth1994}. In this
unrestricted form the Hilbert-space dimension $d$ is not fixed in advance.

The Lov\'{a}sz sandwich theorem~\cite{Lovasz1979,Knuth1994} states that
\begin{equation}
  \al(G)\leq\vt(G)\leq\chi(\bar{G}),
  \label{eq:lovasz_sandwich}
\end{equation}
where $\al(G)$ is the independence number, namely the largest number of
vertices in $G$ with no edges between them, $\bar{G}$ is the complement of
$G$, and $\chi(\bar{G})$ is the chromatic number of the complement, i.e.,
the minimum number of colors required to color its vertices so that adjacent
vertices receive different colors. The relaxation over real vectors is tight:
$\vt(G)$ equals the maximum in Eq.~\eqref{eq:orth_rep} even when the
orthonormal representation is restricted to real vectors
$|v_i\rangle\in\mathbb{R}^d$~\cite{Lovasz1979,Knuth1994}.

\subsection{Quantum Dimension Hierarchy}
\label{sec:dim_hierarchy}

The graph-theoretic quantities introduced above have a direct measurement
interpretation.  To each vertex \(i\) we associate a rank-one projector
\begin{equation}
\Pi_i = |v_i\rangle\langle v_i|,
\end{equation}
representing a two-outcome measurement event.  We work throughout with
rank-one projectors; contextuality scenarios built from projectors of higher
rank have been studied separately~\cite{Xu2021}.  For a quantum state
\(\rho\), the probability of this event is
\begin{equation}
p_i=\mathrm{Tr}(\rho \Pi_i).
\end{equation}
If two vertices \(i\) and \(j\) are connected by an edge, the corresponding
events are exclusive.  In the projective realization used here this is
implemented by orthogonality,
\begin{equation}
\langle v_i|v_j\rangle=0,
\qquad (i,j)\in E,
\end{equation}
so that the two corresponding outcomes cannot occur simultaneously.

For a given graph realization, the contextuality witness considered in this
work is the sum of the event probabilities,
\begin{equation}
S=\sum_{i=1}^{n}p_i
  =\mathrm{Tr}\!\left(\rho\sum_{i=1}^{n}\Pi_i\right).
\end{equation}
The noncontextual bound is \(\alpha(G)\), the independence number introduced
above.  Therefore a quantum realization demonstrates contextuality whenever
\begin{equation}
S>\alpha(G).
\end{equation}
In this sense, the physically relevant ``quantum advantage'' in the present
work is the amount by which the maximum quantum value of \(S\) exceeds the
noncontextual bound.  To avoid ambiguity with computational or metrological
notions of quantum advantage, we will refer to this quantity below as the
\emph{contextuality margin}.

For fixed projectors $\{\Pi_i\}$, define the positive semidefinite operator
\begin{equation}
M=\sum_{i=1}^{n}\Pi_i.
\end{equation}
The state that maximizes $S=\mathrm{Tr}(\rho M)$ is an eigenstate of $M$
corresponding to its largest eigenvalue. Hence
\begin{equation}
\max_{\rho} S = \lambda_{\max}(M),
\end{equation}
where $\lambda_{\max}(M)$ denotes the largest eigenvalue of $M$. This
observation motivates the dimension-restricted quantum value
\begin{equation}
\eta_d(G)=
\max_{\substack{
\{|v_i\rangle\}\subset\mathbb{R}^{d}\\
\langle v_i|v_i\rangle=1\\
\langle v_i|v_j\rangle=0\;\forall(i,j)\in E}}
\lambda_{\max}\!\left(
\sum_{i=1}^{n}|v_i\rangle\langle v_i|
\right).
\label{eq:eta_d}
\end{equation}
Thus, $\eta_d(G)$ is the largest value of the contextuality witness that
can be attained using a $d$-dimensional realization within the real
orthogonal-representation framework employed in our numerical analysis.

It is useful to distinguish this dimension-restricted quantum value from
the corresponding contextuality violation. We therefore define the
dimension-dependent contextuality margin as
\begin{equation}
\Delta_d(G)
=
\eta_d(G)-\alpha(G).
\label{eq:delta_d}
\end{equation}
The condition
\begin{equation}
\Delta_d(G)>0
\end{equation}
means that a $d$-dimensional quantum system is sufficient to exceed the
noncontextual bound and hence demonstrate contextuality for the graph.

This quantity should be distinguished from the unrestricted
graph-theoretic contextuality gap,
\begin{equation}
\Delta(G)
=
\vartheta(G)-\alpha(G),
\label{eq:delta_full}
\end{equation}
where $\vartheta(G)$ is the Lov\'asz theta number. Since
$\eta_d(G)\leq\vartheta(G)$, one has
\begin{equation}
\Delta_d(G)\leq\Delta(G).
\end{equation}
Equality is reached when dimension $d$ is sufficient to realize the full
Lov\'asz value, $\eta_d(G)=\vartheta(G)$. Thus, $\Delta_d(G)$ quantifies
the contextuality accessible in a specified Hilbert-space dimension,
whereas $\Delta(G)$ characterizes the maximum contextuality gap allowed
by the graph without imposing that dimensional restriction.

The KCBS scenario provides a useful physical example.  Its graph is the
pentagon \(C_5\), whose five vertices represent five yes/no measurement
events on a single three-level quantum system.  Each event can be represented
by a rank-one projector
\(\Pi_i=|v_i\rangle\langle v_i|\) in a qutrit Hilbert space, and adjacent
vertices correspond to orthogonal projectors.  Experimentally, each
\(\Pi_i\) asks whether the qutrit occupies a particular one-dimensional
state, with outcomes associated with \(\Pi_i\) and \(I-\Pi_i\).
For the pentagon, the noncontextual bound is
\begin{equation}
\alpha(C_5)=2,
\end{equation}
whereas an appropriate qutrit state and set of projectors give
\begin{equation}
\sum_{i=1}^{5}\langle \Pi_i\rangle=\sqrt{5}.
\end{equation}
The KCBS contextuality margin is therefore
\begin{equation}
\Delta_{3}^{\mathrm{(KCBS)}}
=
\sqrt{5}-2.
\end{equation}
This example illustrates the role of both the graph and the Hilbert-space
dimension: the five graph vertices label five measurement events, while
\(d=3\) specifies the dimension of the quantum system in which those
measurements are realized.  The number of vertices and the Hilbert-space
dimension are therefore distinct quantities.

Increasing the Hilbert-space dimension enlarges the set of available
orthogonal representations. Any \(d\)-dimensional realization can be
embedded in \(d+1\) dimensions, and consequently
\begin{equation}
\eta_d(G)\leq \eta_{d+1}(G)\leq \vartheta(G),
\end{equation}
or equivalently,
\begin{equation}
\Delta_d(G)\leq \Delta_{d+1}(G)\leq \Delta(G).
\end{equation}
When the dimension restriction is sufficiently relaxed, the unrestricted
Lov\'asz value is recovered,
\begin{equation}
\lim_{d\rightarrow\infty}\eta_d(G)=\vartheta(G),
\end{equation}
and hence
\begin{equation}
\lim_{d\rightarrow\infty}\Delta_d(G)=\Delta(G).
\end{equation}
For a finite graph, this saturation is reached in some finite Hilbert-space
dimension.

This distinction is central to the present work.  A graph may already be
contextual for a qutrit while requiring a larger Hilbert space to attain its
full Lov\'asz value.  We therefore define the minimum dimension in which
contextuality can occur as
\begin{equation}
d^{*}(G)
=
\min\left\{
d:\eta_d(G)>\alpha(G)
\right\}.
\label{eq:dstar}
\end{equation}
Within the rank-one framework considered here, \(d^{*}=3\) means that a
single qutrit is sufficient to realize a contextuality violation.  A graph
with no violation in three dimensions but with
\(\eta_4(G)>\alpha(G)\) requires at least the four-dimensional realization
identified by our search.  Such a four-dimensional Hilbert space may be
implemented, for example, using a four-level system or the state space of
two qubits; contextuality itself does not require entanglement.

For the numerical dimension classification below, Eq.~\eqref{eq:eta_d} is
optimized over real vectors, consistently with our computational procedure.
Because this optimization is non-convex, finding a value
\(\eta_d(G)>\alpha(G)\) gives a constructive demonstration of contextuality
in dimension \(d\), whereas failure to find such a value does not by itself
constitute a rigorous upper bound.  We therefore distinguish between
constructively established qutrit witnesses and cases in which the minimum
dimension is only numerically indicated.  For Quad-C5, the three-dimensional
value is later established analytically, including for complex
three-dimensional representations.

\section{Methods}

Our analysis combines exhaustive graph enumeration with numerical optimization
of the corresponding classical and quantum bounds.  For each graph, we first
determine the noncontextual bound from its independence number
$\alpha(G)$.  We then compute the unrestricted projective quantum bound
$\vartheta(G)$ using semidefinite programming (SDP), as introduced in
Sec.~\ref{sec:theta_sdp}.  Comparing these quantities gives the contextuality gap
$\Delta(G)=\vartheta(G)-\alpha(G)$.  After identifying the leading
eight-vertex candidates, we examine which Hilbert-space dimensions are
sufficient to realize their contextuality by optimizing the
dimension-restricted quantity $\eta_d(G)$ defined in Sec.~\ref{sec:dim_hierarchy}.
The unrestricted SDP calculations are convex and can therefore be solved
reliably to high numerical precision, whereas the fixed-dimension
optimization is non-convex and requires a separate multi-start numerical
procedure.  The following subsections describe these steps and the associated
numerical checks.

\subsection{Exhaustive graph enumeration}

To determine the strongest eight-event witnesses according to
$\Delta(G)$, we perform an exhaustive search over all simple connected graphs
with eight vertices.  Here ``non-isomorphic'' means that two graphs that differ
only by a relabelling of their vertices are counted as the same graph.  This
avoids repeatedly analyzing physically equivalent exclusivity structures.

We use McKay's catalogue of non-isomorphic simple graphs~\cite{McKay2014},
which contains all $12\,346$ graphs with $n=8$ vertices. After restricting
to connected graphs, $11\,117$ distinct graphs remain. For each graph, we
compute the noncontextual bound $\alpha(G)$ and the Lov\'asz theta number
$\vartheta(G)$, and rank the candidates according to the contextuality gap
$\Delta(G)$ defined in Sec.~\ref{sec:background}.

Restricting the main search to connected graphs is natural for the present
purpose because a disconnected graph decomposes into independent
exclusivity components rather than representing a genuinely connected
eight-event contextuality structure.

For comparison with smaller canonical scenarios, we also analyze all
non-isomorphic connected graphs with $n=5$ and $n=7$ vertices using the
Graph Atlas implemented in NetworkX~\cite{Hagberg2008}.  These calculations
recover the known five-cycle and seven-cycle results and therefore provide
useful checks of the numerical pipeline before applying it to the complete
eight-vertex set.

Although related graph catalogues and contextuality searches have appeared
previously~\cite{Amselem2012,Canas2016,McKay2014}, we perform the
enumeration and optimization independently because our objective function,
the absolute contextuality gap $\Delta(G)$, differs from the relative
measures considered in some earlier studies.  This also allows us to apply
the same numerical procedure and accuracy checks consistently to every
candidate.

\subsection{Independence number}

For each graph $G$, we compute the independence number $\alpha(G)$, which
sets the noncontextual bound for the contextuality witness considered here.
An independent set is a set of vertices containing no adjacent pair, and
$\alpha(G)$ is the largest possible size of such a set.

Equivalently,
\[
\alpha(G)=\omega(\bar{G}),
\]
where $\bar{G}$ is the complement of $G$ and $\omega(\bar{G})$ denotes its
clique number, namely the number of vertices in its largest complete
subgraph.  This relation follows because vertices that are mutually
nonadjacent in $G$ become mutually adjacent in $\bar{G}$.

We determine $\alpha(G)$ using a maximum-clique search on the complement
graph, implemented in NetworkX.  For the eight-vertex graphs studied here,
we independently verify every result by direct enumeration of all
$2^8=256$ vertex subsets and selecting the largest subset containing no
internal edge.  The two procedures agree for all $11\,117$ connected
eight-vertex graphs, providing an exact cross-check of the noncontextual
bounds used in the subsequent analysis.

\subsection{Lov\'asz theta number from semidefinite programming}
\label{sec:method_sdp}

For each graph, we evaluate the Lov\'asz theta number $\vartheta(G)$ by
solving the semidefinite program in Eq.~\eqref{eq:sdp}.  Semidefinite
programming (SDP) is a convex optimization method in which the matrix
variable is constrained to remain positive semidefinite.  Because the
optimization is convex, the unrestricted quantum bound
$\vartheta(G)$ can be determined reliably to high numerical accuracy.

We implement the semidefinite program in CVXPY~\cite{Diamond2016}, a Python framework for convex optimization, and solve it using the Splitting Conic Solver (SCS)~\cite{ODonoghue2016} and CLARABEL~\cite{Goulart2026} conic numerical solver packages.
The complete set
of $11{,}117$ connected eight-vertex graphs is first screened using the
SCS.  This stage provides
sufficient numerical accuracy to identify the graphs with the largest
contextuality gaps.  We then recompute the 50 leading candidates using the
higher-accuracy CLARABEL conic solver and use these refined values in the
reported rankings and tables.

As a numerical consistency check, the two independent solvers give
Lov\'asz-theta values that agree to better than $10^{-7}$ for every graph
in the leading set.  We also monitor the residuals of the SDP constraints.
These residuals quantify how accurately the numerical solution satisfies,
for example, the trace normalization and the zero-matrix-element conditions
associated with graph edges.  For the high-accuracy calculations, the
largest relevant feasibility residual is below $4\times10^{-10}$.

Because our main conclusion depends on the ordering of the leading graphs
according to $\Delta(G)$, we tested the numerical stability of their
Lov\'asz values using repeated high-accuracy calculations with two
independent solvers.  For the leading candidates, the two solvers agree to
better than $10^{-7}$, and the resulting numerical ranges are well
separated: in particular, the lowest value obtained for \QC\ remains above
the highest value obtained for the second-ranked graph by
$2.94\times10^{-2}$, many orders of magnitude larger than the observed
solver residuals.  These checks support the identification of \QC\ as the
unique eight-vertex maximizer of $\Delta(G)$ at the reported numerical
precision.  This remains a high-precision numerical result rather than an
exact mathematical certificate; a fully rigorous proof would require exact
arithmetic or independently verified interval bounds for the SDP.
Detailed numerical ranges and solver diagnostics are given in
Tables~\ref{tab:intervals} and \ref{tab:certification} of
Appendix~\ref{app:tables}.

\subsection{Dimension classification via constrained optimization}
\label{sec:dimension_classification}

Determining $\eta_d(G)$ at fixed Hilbert-space dimension is a non-convex
optimization problem.  Unlike the semidefinite program used to obtain
$\vartheta(G)$, the fixed-dimension search may contain local optima.
Accordingly, any value found by the numerical optimization provides a
constructive lower bound on $\eta_d(G)$.

We use a three-stage optimization procedure with 300 random initial
configurations for each graph and dimension considered.  First, the vectors
are adjusted to minimize violations of the required orthogonality relations,
\begin{equation}
L_{\perp}
=
\sum_{(i,j)\in E}
\left|
\langle \hat v_i|\hat v_j\rangle
\right|^2 .
\end{equation}
Second, we maximize
\begin{equation}
\lambda_{\max}\!\left(
\sum_i |v_i\rangle\langle v_i|
\right)
\end{equation}
while penalizing deviations from normalization and orthogonality.  Finally,
the best candidates are refined using an optimization with the normalization
and orthogonality conditions imposed as equality constraints.  We use
standard constrained numerical routines for these stages.
The first two stages use L-BFGS-B optimization, while the final constrained
refinement is performed using sequential least-squares programming (SLSQP).
The penalty weights used throughout are $w_{\perp}=5.0\times10^{4}$,
$w_{\lambda}=0.3$, and $w_{\mathrm{norm}}=5.0\times10^{3}$.

The largest value obtained over all random restarts is taken as our numerical
estimate of $\eta_d(G)$.  If this value exceeds $\alpha(G)$, then
\begin{equation}
\Delta_d(G)>0,
\end{equation}
which constructively demonstrates contextuality in dimension $d$.

If, instead, all restarts return a best value equal to the noncontextual
bound, we report that no $d$-dimensional contextuality violation was found.
This result is consistent with $\Delta_d(G)=0$, but it does not constitute
a rigorous upper bound on $\eta_d(G)$ because the optimization is non-convex.
A formal exclusion of $\Delta_d(G)>0$ would require an independent upper-bound
method, such as an appropriate semidefinite relaxation.

\section{Results}
\label{sec:results}

\subsection{Reference results: $n=5$ and $n=7$}

Before turning to the eight-vertex search, we apply the same analysis to the
smaller odd-cycle scenarios that provide standard reference points for
contextuality.  For $n=5$, the maximum contextuality gap is attained by the
five-cycle $C_5$, the exclusivity graph of the KCBS scenario
~\cite{Klyachko2008}:
\begin{equation}
  \alpha(C_5)=2,\quad
  \vartheta(C_5)=\sqrt{5},\quad
  \Delta(C_5)=\sqrt{5}-2\approx0.236.
\end{equation}

For $n=7$, the maximum is attained by the seven-cycle
$C_7$~\cite{Araujo2013}.  An exhaustive semidefinite-programming search over
all 853 connected non-isomorphic seven-vertex graphs reproduces
\begin{equation}
  \alpha(C_7)=3,\qquad
  \vartheta(C_7)
  =\frac{7\cos(\pi/7)}{1+\cos(\pi/7)}
  \approx3.318,
\end{equation}
and therefore
\begin{equation}
  \Delta(C_7)\approx0.318.
\end{equation}

Both graphs are qutrit witnesses,
\begin{equation}
  d^{*}(C_5)=d^{*}(C_7)=3.
\end{equation}
These results provide two useful benchmarks for the eight-vertex problem.
They verify that our numerical pipeline reproduces the known odd-cycle
results and illustrate that a larger measurement scenario can support a
larger contextuality gap while remaining realizable in a three-dimensional
Hilbert space.  The eight-vertex search below asks whether this trend
continues beyond the analytically characterized cycle family and, in
particular, whether a different graph structure can yield a still larger
gap.

\subsection{Exhaustive search over $n=8$}

The exhaustive search identifies \QC\ as the connected eight-vertex graph
with the largest contextuality gap, with
\begin{equation}
\Delta(\QC)=0.46784.
\end{equation}
The leading eight-vertex candidates are summarized in
Table~\ref{tab:n8_top10}. Complete edge lists for the six leading graphs
are provided in Table~\ref{tab:edgelists} of Appendix~\ref{app:tables}
for reproducibility.
Of
particular interest is the comparison with the Wagner graph $W$, a standard
eight-vertex benchmark: \QC\ supports a larger separation between the
quantum and noncontextual bounds while requiring only 10 edges, compared
with 12 for $W$.

This result shows that a stronger contextuality witness does not necessarily
require a denser network of exclusivity relations.  Instead, the way in
which the measurement events and their mutual exclusivities are organized
is crucial.  The structure responsible for this enhanced gap is examined
in the following subsection.

\begin{table}[tp]
  \caption{Top-10 connected graphs on $n=8$ vertices by contextuality gap
           $\De=\vt-\al$ (CLARABEL values), from exhaustive search over the
           11{,}117 connected non-isomorphic simple 8-vertex graphs.
           The column $\vt/\al$ is the relative contextuality measure
           used in earlier graph searches~\cite{Canas2016,Ray2021}; its $n{=}8$
           maximizer (Rank~7, $\dagger$, the CHSH-type graph discussed
           in Ref.~\cite{Canas2016})
           differs from the gap maximizer (\QC, Rank~1).}
  \label{tab:n8_top10}
  \begin{ruledtabular}
  \scriptsize
  \begin{tabular}{clrrrrr}
    Rank & Note & \multicolumn{1}{c}{$|E|$} &
          \multicolumn{1}{c}{$\al$} &
          \multicolumn{1}{c}{$\vt$} &
          \multicolumn{1}{c}{$\De$} &
          \multicolumn{1}{c}{$\vt/\al$} \\
    \hline
    1  & \QC\ ($n{=}8$ gap max)        & 10 & 3 & 3.46784 & 0.46784 & 1.15595 \\
    2  &                              & 11 & 3 & 3.43845 & 0.43845 & 1.14615 \\
    3  & Wagner $W$                   & 12 & 3 & 3.41421 & 0.41421 & 1.13807 \\
    4  &                              & 11 & 3 & 3.37228 & 0.37228 & 1.12409 \\
    5  &                              & 12 & 3 & 3.37228 & 0.37228 & 1.12409 \\
    6  &                              & 10 & 3 & 3.37228 & 0.37228 & 1.12409 \\
    7  & $\dagger$ CHSH-type graph     & 16 & 2 & 2.34315 & 0.34315 & 1.17157 \\
    8  &                              & 11 & 3 & 3.33804 & 0.33804 & 1.11268 \\
    9  &                              & 17 & 2 & 2.33804 & 0.33804 & 1.16902 \\
    10 &                              & 18 & 2 & 2.33333 & 0.33333 & 1.16667 \\
  \end{tabular}
  \end{ruledtabular}
\end{table}

\begin{figure}[tp]
\centering
\begin{tikzpicture}[scale=0.88]

\begin{scope}[xshift=-2.2cm]
  \node[hub]  (v3) at (-1.5,  1.8) {3};
  \node[leaf] (v0) at (-1.5,  0.6) {0};
  \node[leaf] (v5) at (-1.5, -0.6) {5};
  \node[hub]  (v2) at (-1.5, -1.8) {2};
  \node[hub]  (v6) at ( 1.5,  1.8) {6};
  \node[leaf] (v1) at ( 1.5,  0.6) {1};
  \node[leaf] (v4) at ( 1.5, -0.6) {4};
  \node[hub]  (v7) at ( 1.5, -1.8) {7};
  \draw[gedge] (v3)--(v0)--(v5)--(v2);
  \draw[gedge] (v6)--(v1)--(v4)--(v7);
  \draw[gedge] (v3)--(v6);
  \draw[gedge] (v2)--(v7);
  \draw[gedge] (v3)--(v7);
  \draw[gedge] (v2)--(v6);
  \node[below=5pt, align=center, font=\footnotesize] at (0,-2.15)
    {(a)~\QC\\[1pt]
     $n{=}8,\;|E|{=}10,\;\De{=}0.4678$};
\end{scope}

\begin{scope}[xshift=2.4cm]
  \def\R{1.8}
  \foreach \k/\ang in {0/90,1/45,2/0,3/315,4/270,5/225,6/180,7/135}{
    \node[hub] (w\k) at (\ang:\R) {\k};
  }
  \draw[gedge] (w0)--(w1)--(w2)--(w3)--(w4)--(w5)--(w6)--(w7)--(w0);
  \draw[diagedge] (w0)--(w4) (w1)--(w5) (w2)--(w6) (w3)--(w7);
  \node[below=5pt, align=center, font=\footnotesize] at (0,-\R-0.2)
    {(b)~Wagner $W$\\[1pt]
     $n{=}8,\;|E|{=}12,\;\De{=}0.4142$};
\end{scope}

\end{tikzpicture}

\caption{
(a) The \QC\ exclusivity graph identified by the exhaustive eight-vertex
search, with $n=8$, $|E|=10$, and
$\Delta(\QC)\approx0.4678$.  Four vertices have degree~2 and four have
degree~3.  (b) The Wagner graph $W$, with $n=8$, $|E|=12$, and
$\Delta(W)\approx0.4142$.  Despite having two fewer exclusivity relations,
\QC\ exhibits the larger contextuality gap.}
\label{fig:graphs}
\end{figure}

\newcommand{\drawQCsmall}[2]{%
\begin{tikzpicture}[
  vxB/.style={circle,fill=black!20,draw=black!45,
              minimum size=14pt,inner sep=0pt,
              font=\scriptsize},
  vxH/.style={circle,fill=black!75,draw=black,
              minimum size=14pt,inner sep=0pt,
              font=\scriptsize\bfseries\color{white}},
  egB/.style={thick,black!20},
  egH/.style={very thick,black!85},
  scale=0.62
]
\coordinate (v3) at (-1.5, 1.8);
\coordinate (v0) at (-1.5, 0.6);
\coordinate (v5) at (-1.5,-0.6);
\coordinate (v2) at (-1.5,-1.8);
\coordinate (v6) at ( 1.5, 1.8);
\coordinate (v1) at ( 1.5, 0.6);
\coordinate (v4) at ( 1.5,-0.6);
\coordinate (v7) at ( 1.5,-1.8);

\draw[egB] (v3)--(v0)--(v5)--(v2);
\draw[egB] (v6)--(v1)--(v4)--(v7);
\draw[egB] (v3)--(v6);
\draw[egB] (v2)--(v7);
\draw[egB] (v3)--(v7);
\draw[egB] (v2)--(v6);

#1

\foreach \v/\lab in
{v0/0,v1/1,v2/2,v3/3,v4/4,v5/5,v6/6,v7/7}
{\node[vxB] at (\v) {\lab};}

#2
\end{tikzpicture}%
}

\begin{figure*}[t]
\centering

\begin{minipage}{0.23\textwidth}
\centering
\drawQCsmall{
  \draw[egH] (v0)--(v3)--(v6)--(v2)--(v5)--(v0);
}{
  \node[vxH] at (v0) {0};
  \node[vxH] at (v3) {3};
  \node[vxH] at (v6) {6};
  \node[vxH] at (v2) {2};
  \node[vxH] at (v5) {5};
}
\\[-1mm]
{\small (a) $C_5^{(1)}$}
\end{minipage}
\hfill
\begin{minipage}{0.23\textwidth}
\centering
\drawQCsmall{
  \draw[egH] (v0)--(v3)--(v7)--(v2)--(v5)--(v0);
}{
  \node[vxH] at (v0) {0};
  \node[vxH] at (v3) {3};
  \node[vxH] at (v7) {7};
  \node[vxH] at (v2) {2};
  \node[vxH] at (v5) {5};
}
\\[-1mm]
{\small (b) $C_5^{(2)}$}
\end{minipage}
\hfill
\begin{minipage}{0.23\textwidth}
\centering
\drawQCsmall{
  \draw[egH] (v1)--(v4)--(v7)--(v2)--(v6)--(v1);
}{
  \node[vxH] at (v1) {1};
  \node[vxH] at (v4) {4};
  \node[vxH] at (v7) {7};
  \node[vxH] at (v2) {2};
  \node[vxH] at (v6) {6};
}
\\[-1mm]
{\small (c) $C_5^{(3)}$}
\end{minipage}
\hfill
\begin{minipage}{0.23\textwidth}
\centering
\drawQCsmall{
  \draw[egH] (v1)--(v4)--(v7)--(v3)--(v6)--(v1);
}{
  \node[vxH] at (v1) {1};
  \node[vxH] at (v4) {4};
  \node[vxH] at (v7) {7};
  \node[vxH] at (v3) {3};
  \node[vxH] at (v6) {6};
}
\\[-1mm]
{\small (d) $C_5^{(4)}$}
\end{minipage}

\caption{The four KCBS-type five-cycles contained in the \QC\ graph.
In each panel, the vertices and exclusivity relations belonging to one
pentagon are highlighted while the remainder of the graph is shown in the
background.  Every edge of \QC\ participates in exactly two of the four
pentagons.}
\label{fig:c5_subgraphs}
\end{figure*}
\subsection{Structure of the \QC\ graph}
\label{sec:qc_structure}

Figure~\ref{fig:graphs}(a) shows the \QC\ graph identified by the
eight-vertex search.  Its eight vertices represent measurement events and
its 10 edges represent pairs of mutually exclusive events.  With the
labelling used in Fig.~\ref{fig:graphs}(a), the graph is specified by
\begin{align}
V&=\{0,1,\ldots,7\},\\
E&=\{(0,3),(0,5),(1,4),(1,6),(2,5),(2,6),
\nonumber\\
&\hspace{1.5cm}(2,7),(3,6),(3,7),(4,7)\}.
\label{eq:qc_edges}
\end{align}
For computational reproducibility, the canonical \texttt{graph6} encoding
of \QC\ is \texttt{GCQb`o}~\cite{McKay2014}.

Four vertices participate in two exclusivity relations and four participate
in three, giving the degree sequence
\[
(2,2,2,2,3,3,3,3),
\]
where the degree of a vertex is simply the number of exclusivity relations
involving the corresponding measurement event.  Thus, the enhanced
contextuality gap found in the preceding subsection is obtained from a
relatively sparse set of only 10 pairwise constraints.

The key structural feature of \QC\ is more informative than the degree
sequence itself.  The graph contains exactly four distinct five-event cycles
of the KCBS type:
\begin{align}
C_5^{(1)} &: 0\text{-}3\text{-}6\text{-}2\text{-}5\text{-}0, \nonumber\\
C_5^{(2)} &: 0\text{-}3\text{-}7\text{-}2\text{-}5\text{-}0, \nonumber\\
C_5^{(3)} &: 1\text{-}4\text{-}7\text{-}2\text{-}6\text{-}1, \nonumber\\
C_5^{(4)} &: 1\text{-}4\text{-}7\text{-}3\text{-}6\text{-}1.
\label{eq:qc_pentagons}
\end{align}
These four pentagons are illustrated separately in
Fig.~\ref{fig:c5_subgraphs}.  Remarkably, every edge of \QC\ belongs to
exactly two of them.  The eight-event graph therefore combines four
overlapping KCBS structures in a highly economical way: the same
exclusivity relations are repeatedly shared among different pentagonal
contexts.

This overlapping pentagonal structure is central to the physical
interpretation of \QC.  In three dimensions, the orthogonality constraints
force part of the eight-event realization to reduce to a KCBS-type
structure, leading to a qutrit contextuality margin closely connected to
the original KCBS result.  In four dimensions this restriction is lifted,
allowing the graph to attain its larger full contextuality gap.  The
dimension dependence is analyzed below, and the exact qutrit result is
established in Sec.~\ref{sec:eta3_exact}.

\subsection{Dimension dependence of contextuality}
\label{sec:dim_class}

The graph-theoretic gap $\Delta(G)$ tells us the largest separation between
quantum and noncontextual predictions allowed by a graph, but it does not
specify the Hilbert-space dimension required to realize that separation.
We therefore next ask a physically distinct question: which of the leading
eight-vertex graphs can already demonstrate contextuality in a qutrit, and
which require a higher-dimensional quantum system?

For a three-dimensional realization, contextuality is present whenever
\[
\Delta_3(G)=\eta_3(G)-\alpha(G)>0.
\]
Our constrained optimization finds such a violation for the graphs ranked
1, 2, 4, and 6 in the eight-vertex search. These graphs are therefore
constructively established as qutrit contextuality witnesses. In particular,
\QC\ is already contextual in three dimensions, even though its full
graph-theoretic contextuality gap is larger than the value accessible to a
qutrit.

For \QC, the qutrit value is
\begin{equation}
\eta_3(\QC)=1+\sqrt{5},
\end{equation}
which, since $\alpha(\QC)=3$, gives
\begin{equation}
\Delta_3(\QC)
=
\sqrt{5}-2
\approx0.236.
\label{eq:qc_delta3}
\end{equation}
This is precisely the contextuality margin of the KCBS qutrit scenario.
The equality is not merely numerical: an analytical derivation is given in
Sec.~\ref{sec:eta3_exact}.  The larger unrestricted gap
$\Delta(\QC)\approx0.468$ is reached only when the realization is extended
to four dimensions.  Thus, the same eight-event graph exhibits two distinct
regimes: KCBS-sized contextuality in a qutrit and a stronger contextuality
gap when an additional Hilbert-space dimension becomes available.

The Wagner graph shows a contrasting behavior.  In all of our
three-dimensional searches, the largest value found is equal to its
noncontextual bound,
\[
\eta_3(W)=3
\]
to the numerical precision of the optimization, so that no qutrit
contextuality violation is found.  In four dimensions, by contrast, the
optimization yields a value above the noncontextual bound.  The numerical
results are therefore consistent with the Wagner witness requiring a
four-dimensional realization.

This dimensional assignment for the Wagner graph should be interpreted with
the qualification discussed in Sec.~\ref{sec:dimension_classification}.
Because the fixed-dimension optimization is non-convex, failure to find a
three-dimensional violation does not constitute a rigorous proof that none
exists.  A formal demonstration that $d^*(W)=4$ would require an independent
upper bound on $\eta_3(W)$.  This limitation does not affect the qutrit
classification of \QC, whose three-dimensional value is established
analytically below.

Interestingly, the numerical value $1+\sqrt{5}$ also appears for the
three-dimensional realizations of the graphs ranked 2 and 4.  For these
graphs it remains a numerical observation, whereas for \QC\ the value is
proved exactly.  This repeated appearance suggests that KCBS-type
five-cycle structure may play a broader role among the leading
eight-vertex qutrit witnesses.

The complete dimension-restricted numerical results for the leading graphs
are summarized in Table~\ref{tab:dim_diagnostics}(a) of
Appendix~\ref{app:tables}.

The four-dimensional realization reveals the complementary part of this
dimension dependence.  For \QC, the numerical optimization gives
\begin{equation}
\eta_4(\QC)=3.467844,
\end{equation}
which agrees with the unrestricted Lov\'asz value
$\vartheta(\QC)$ to the numerical precision of the calculation.  Hence
\begin{equation}
\Delta_4(\QC)=\Delta(\QC),
\end{equation}
showing that four dimensions are sufficient to access the full
contextuality gap of the graph.

The numerical optimization also provides a useful consistency check on
these dimension-restricted realizations.  For \QC\ in both $d=3$ and
$d=4$, the optimized vectors satisfy the required orthogonality relations
to better than $10^{-13}$, and repeated random initializations converge to
the same optimum. These numerical diagnostics are summarized in
Table~\ref{tab:dim_diagnostics}(b) of Appendix~\ref{app:tables}.

The resulting physical picture is therefore simple.  The same eight-event
graph is already contextual for a qutrit, where its contextuality margin is
KCBS-sized, while a four-dimensional realization accesses the substantially
larger contextuality gap allowed by the full graph.

\subsection{Exact qutrit value}
\label{sec:eta3_exact}

The three-dimensional optimization strongly suggests the simple exact value
\begin{equation}
  \eta_3(\QC)=1+\sqrt{5}\approx3.23607.
  \label{eq:eta3_num}
\end{equation}
A low-degree integer-relation search~\cite{Bailey1989} on the numerical
result returns the polynomial $P(x)=x^2-2x-4$, whose positive root is
$1+\sqrt{5}$; the optimized value agrees with it to approximately
$3\times10^{-13}$.  This identification only suggests the exact form and is
not a proof, and the result does not rest on it:
Proposition~\ref{prop:eta3_exact} below establishes it analytically.

We first obtain the lower bound
\begin{equation}
  \eta_3(\QC)\geq1+\sqrt{5}
  \label{eq:eta3_lb}
\end{equation}
from the explicit three-dimensional orthogonal representation given in
Appendix~\ref{app:analytic}.  The matching upper bound follows directly from
the pentagonal structure of \QC\ established in
Sec.~\ref{sec:qc_structure}.  We first record a general bound that will be
used in the proof.

\begin{lemma}
\label{lem:rep_bound}
Let $\{|v_i\rangle\}$ be an orthogonal representation of a graph $G$ by unit
vectors, with $P_i=|v_i\rangle\langle v_i|$. Then
\begin{equation}
  \lambda_{\max}\!\left(\sum_i P_i\right)\leq\vartheta(G).
\end{equation}
\end{lemma}

\begin{proof}
For a unit vector $|\psi\rangle$,
\begin{equation}
  \left\langle\psi\left|\sum_iP_i\right|\psi\right\rangle
  =
  \sum_i|\langle\psi|v_i\rangle|^2 .
\end{equation}
Therefore
\begin{equation}
  \lambda_{\max}\!\left(\sum_iP_i\right)
  =
  \max_{|\psi\rangle}
  \sum_i|\langle\psi|v_i\rangle|^2 .
\end{equation}
Equation~\eqref{eq:orth_rep} maximizes the same quantity over all orthogonal
representations of $G$ and all handle states $|\psi\rangle$.  Since the
chosen $\{|v_i\rangle\}$ is one such representation, its maximum cannot
exceed $\vartheta(G)$.
\end{proof}

\begin{proposition}
\label{prop:eta3_exact}
\begin{equation}
  \eta_3(\QC)=1+\sqrt{5}.
\end{equation}
\end{proposition}

\begin{proof}
The lower bound is given by Eq.~\eqref{eq:eta3_lb}.  For the upper bound,
let $\{|v_i\rangle\}\subset\mathbb{R}^3$ be an arbitrary orthogonal
representation of \QC, and write
$P_i=|v_i\rangle\langle v_i|$, using the vertex labelling of
Fig.~\ref{fig:graphs}(a) and the edge set in Eq.~\eqref{eq:qc_edges}.

Vertices $2$ and $3$ are both adjacent to vertices $6$ and $7$. Hence
\begin{equation}
  \operatorname{span}\{v_2,v_3\}
  \perp
  \operatorname{span}\{v_6,v_7\}.
\end{equation}
Two mutually orthogonal subspaces of $\mathbb{R}^3$ cannot both be
two-dimensional.  At least one of these spans must therefore be
one-dimensional, implying either
\begin{equation}
  P_2=P_3
  \qquad\text{or}\qquad
  P_6=P_7.
\end{equation}
The graph is invariant under the vertex relabelling
$(0\,1)(2\,7)(3\,6)(4\,5)$, which interchanges the pairs
$\{2,3\}$ and $\{6,7\}$ without changing
$\lambda_{\max}(\sum_iP_i)$.  We may therefore take
$P_2=P_3$, or equivalently $v_3=\pm v_2$.

The edge $(0,3)$ implies $v_0\perp v_3$ and therefore
$v_0\perp v_2$.  Together with the edges $(0,5)$ and $(2,5)$,
the vectors $v_0$, $v_2$, and $v_5$ are pairwise orthogonal unit vectors
in $\mathbb{R}^3$.  They therefore form an orthonormal basis, giving
\begin{equation}
  P_0+P_2+P_5=I,
  \label{eq:basis_triple}
\end{equation}
where $I$ is the identity operator in the qutrit Hilbert space.

The five remaining vertices $\{1,3,4,6,7\}$ form one of the KCBS pentagons
identified in Eq.~\eqref{eq:qc_pentagons} and
Fig.~\ref{fig:c5_subgraphs}.  Their projectors therefore constitute an
orthogonal representation of $C_5$.  Lemma~\ref{lem:rep_bound} gives
\begin{equation}
  \lambda_{\max}
  \left(P_1+P_3+P_4+P_6+P_7\right)
  \leq
  \vartheta(C_5)
  =
  \sqrt{5}.
\end{equation}
Using Eq.~\eqref{eq:basis_triple}, we then obtain
\begin{align}
  \lambda_{\max}\!\left(\sum_{i=0}^{7}P_i\right)
  &=
  \lambda_{\max}
  \left(
  I+P_1+P_3+P_4+P_6+P_7
  \right)
  \nonumber\\
  &\leq1+\sqrt{5}.
\end{align}
Since the original three-dimensional representation was arbitrary, this is
a general upper bound.  Together with Eq.~\eqref{eq:eta3_lb}, it proves
$\eta_3(\QC)=1+\sqrt{5}$.
\end{proof}

The proof uses only orthogonality relations and the three-dimensionality of
the Hilbert space, and therefore extends directly to
$\mathbb{C}^3$.  In that case the one-dimensional condition gives
$v_3=e^{i\phi}v_2$ rather than $v_3=\pm v_2$, but the corresponding
projectors still satisfy $P_2=P_3$, so the remainder of the argument is
unchanged.

The exact result has a simple physical interpretation.  Since
$\alpha(\QC)=3$, the qutrit contextuality margin is
\begin{equation}
  \Delta_3(\QC)
  =
  \eta_3(\QC)-\alpha(\QC)
  =
  \sqrt{5}-2.
\end{equation}
Thus
\begin{equation}
  \Delta_3(\QC)
  =
  \Delta_3^{\mathrm{(KCBS)}},
\end{equation}
showing that the eight-event \QC\ graph has exactly the same contextuality
margin as the KCBS pentagon when both are restricted to a qutrit.  The
larger contextuality gap identified by the eight-vertex search therefore
arises only when the additional Hilbert-space dimension is available, as
discussed in Sec.~\ref{sec:dim_class}.

Taken together, the unrestricted contextuality gaps of the standard
reference graphs considered here follow the hierarchy
\begin{equation}
  \Delta(C_5)
  <
  \Delta(C_7)
  <
  \Delta(W)
  <
  \Delta(\QC).
\end{equation}
Thus, the progression from the minimal KCBS pentagon to larger
measurement scenarios allows increasingly large quantum--noncontextual
separations, with \QC\ giving the largest gap among the connected
eight-vertex graphs considered here.  As shown above, however, this
ordering of the unrestricted gaps does not by itself determine the
Hilbert-space dimension required to realize them: \QC\ is already
contextual for a qutrit, while its full gap is reached only in four
dimensions.

\section{Discussion}

\subsection{Sparse structure and contextuality gap}

One noteworthy feature of \QC\ is that its large contextuality gap is
obtained from a relatively sparse exclusivity structure.  The graph contains
only 10 edges, compared with 12 for the Wagner graph, yet
\begin{equation}
\Delta(\QC)\approx0.468
>
\Delta(W)\approx0.414.
\end{equation}
Thus, increasing the number of pairwise exclusivity constraints does not by
itself increase the separation between the quantum and noncontextual bounds.

This behavior can be understood qualitatively from the two quantities that
enter the contextuality gap,
\begin{equation}
\Delta(G)=\vartheta(G)-\alpha(G).
\end{equation}
Adding or removing an edge changes the allowed exclusivity relations and can
therefore affect both the noncontextual bound $\alpha(G)$ and the quantum
bound $\vartheta(G)$.  The resulting change in their difference is not
determined by the number of edges alone.  What matters is how the
exclusivity relations are arranged throughout the graph.

The structure of \QC\ provides a particularly suggestive example.  As shown
in Sec.~\ref{sec:qc_structure} and Fig.~\ref{fig:c5_subgraphs}, its 10 edges
organize eight measurement events into four overlapping KCBS pentagons, with
each edge shared by exactly two of them.  This repeated use of the same
exclusivity relations distinguishes \QC\ from the Wagner graph and provides
a natural structural motif associated with its larger value of
$\Delta(G)$.

The present results do not establish that the two-fold pentagon overlap is,
by itself, the mathematical cause of the enhanced gap.  Rather, the
exhaustive search identifies \QC\ as the optimal eight-vertex graph, while
the subsequent structural analysis shows that this optimum possesses a
highly organized KCBS-type decomposition.  Determining whether such
overlapping odd-cycle structures provide a more general principle for
constructing compact contextuality witnesses is an interesting question for
larger graphs.

\subsection{Physical significance of the Hilbert-space dimension}

The dimension dependence discussed above has a direct experimental meaning.
A witness with $d^*=3$ can demonstrate contextuality using a single
three-level quantum system, or qutrit.  Possible realizations include, for
example, three-dimensional subspaces of photonic orbital angular momentum
~\cite{Molina-Terriza2007} and atomic three-level systems
~\cite{Hu2016}.  The KCBS experiment provides the canonical example:
contextuality can be tested within a single qutrit without requiring
spatially separated subsystems.

The situation changes when a graph requires a four-dimensional Hilbert
space to exhibit contextuality or to attain its full quantum value.
Importantly, $d=4$ specifies the dimension of the state space, not a unique
physical architecture.  Such a realization may use a single four-level
system, four photonic modes, or the four-dimensional joint Hilbert space of
two qubits.  Thus, the distinction between $d=3$ and $d=4$ should be
understood primarily as a dimensional requirement on the available quantum
states and measurements.

For \QC, this distinction is particularly transparent.  A qutrit already
produces contextuality,
\begin{equation}
\Delta_3(\QC)=\sqrt{5}-2,
\end{equation}
but the full graph-theoretic gap is attained only in four dimensions,
\begin{equation}
\Delta_4(\QC)=\Delta(\QC).
\end{equation}
The additional Hilbert-space dimension therefore does not merely provide an
alternative implementation; it makes a larger separation between quantum
and noncontextual predictions accessible.

The Wagner graph provides a contrasting example.  Our numerical search finds
no violation in three dimensions, while a four-dimensional realization does
violate the noncontextual bound.  This behavior is consistent with
$d^*(W)=4$, although, as emphasized in Sec.~\ref{sec:dim_class}, the absence
of a three-dimensional violation is presently a numerical result rather than
a rigorous upper-bound proof.

A four-dimensional realization should also not be identified automatically
with Bell nonlocality.  Contextuality concerns the compatibility and outcome
assignments of measurements and does not require spatially separated
parties.  If the four-dimensional space is encoded using two qubits, the
presence or absence of entanglement depends on the particular state and
measurement realization; the dimensional classification itself makes no
assumption about entanglement.  Contextuality and Bell nonlocality therefore
remain conceptually distinct notions~\cite{Spekkens2005,Porto2024}.

Taken together, these results show that Hilbert-space dimension acts as a
physical resource constraint on the contextuality accessible from a given
exclusivity graph.  \QC\ is especially useful in this respect because the
same eight-event structure supports a KCBS-sized contextuality margin in a
qutrit and a substantially larger contextuality gap when realized in four
dimensions.

\subsection{Robustness to depolarizing noise}
\label{sec:noise}

The contextuality gap $\Delta(G)$ quantifies the separation between the
noncontextual and unrestricted quantum bounds, but it does not by itself
determine how robust a particular realization is to experimental noise.
To obtain a simple operational benchmark, we consider depolarizing noise
acting on the quantum state while keeping the projective measurements ideal,
\begin{equation}
  \rho_v
  =
  v\,\rho+(1-v)\frac{I}{d},
  \qquad
  0\leq v\leq1,
\end{equation}
where $v$ is the visibility.

For a $d$-dimensional realization with rank-one projectors
$P_i=|v_i\rangle\langle v_i|$, define the contextuality witness
\begin{equation}
  S=\sum_{i=1}^{n}\Tr(\rho P_i).
\end{equation}
For the state and projectors attaining $\eta_d(G)$, depolarization gives
\begin{equation}
  S(v)
  =
  v\,\eta_d(G)
  +(1-v)\frac{n}{d},
\end{equation}
because $\Tr(P_i)=1$ for each of the $n$ projectors.  Contextuality remains
observable as long as $S(v)>\alpha(G)$.  The corresponding critical
visibility is therefore
\begin{equation}
  v^*(G,d)
  =
  \frac{\alpha(G)-n/d}
       {\eta_d(G)-n/d},
  \label{eq:vstar}
\end{equation}
for the contextual realizations considered here.  Smaller $v^*$ corresponds
to greater tolerance to depolarizing noise, with the depolarizing fraction
at threshold given by $1-v^*$.

Table~\ref{tab:noise} compares this operational figure of merit for the
reference graphs and for the three- and four-dimensional realizations of
\QC.

\begin{table}[tp]
  \caption{\label{tab:noise}
  Critical visibility $v^*$ and corresponding depolarizing fraction $1-v^*$ for the contextuality witnesses considered here. The quantity $\eta_d(G)$ denotes the quantum value used for the indicated Hilbert-space dimension. Where the dimension-restricted realization attains the unrestricted Lov\'asz value, we have $\eta_d(G)=\vartheta(G)$ at the reported numerical precision.}
  \begin{ruledtabular}
  \begin{tabular}{lrrrrcc}
    Graph & $n$ & $\alpha$ & $\eta_d$ & $d$ & $v^*$ & $1-v^*$ \\
    \hline
    $C_5$ (KCBS)       & 5 & 2 & 2.23607 & 3 & 0.5854 & 0.4146 \\
    $C_7$              & 7 & 3 & 3.31767 & 3 & 0.6773 & 0.3227 \\
    \QC\ ($d=3$)       & 8 & 3 & 3.23607 & 3 & 0.5854 & 0.4146 \\
    \QC\ ($d=4$)       & 8 & 3 & 3.46784 & 4 & 0.6813 & 0.3187 \\
    Wagner $W$ ($d=4$) & 8 & 3 & 3.41421 & 4 & 0.7071 & 0.2929 \\
  \end{tabular}
  \end{ruledtabular}
\end{table}

A particularly simple result emerges for the qutrit realization of \QC.
Using
\begin{equation}
  \eta_3(\QC)=1+\sqrt{5},
  \qquad
  \alpha(\QC)=3,
\end{equation}
together with $n=8$ and $d=3$, Eq.~\eqref{eq:vstar} gives exactly the same
critical visibility as for the KCBS pentagon:
\begin{equation}
  v^*(\QC,3)
  =
  v^*(C_5,3)
  =
  \frac{1}{3\sqrt{5}-5}
  \approx0.5854.
  \label{eq:vstar_equal}
\end{equation}
Thus, within this depolarizing-noise model, both qutrit witnesses remain
contextual up to a depolarizing fraction of approximately $41.5\%$.

This exact equality of the critical visibilities complements the result
$\Delta_3(\QC)=\Delta_3^{\mathrm{(KCBS)}}=\sqrt{5}-2$: within the
state-depolarization model, the qutrit realizations of \QC\ and KCBS have
both the same contextuality margin and the same noise threshold.

The four-dimensional realizations illustrate a different aspect of the same
criterion.  For \QC, we find
$v^*\approx0.6813$, whereas the Wagner graph gives
$v^*=1/\sqrt{2}\approx0.7071$.  The critical visibility of \QC\ is therefore
lower by approximately $0.026$, or $2.6$ percentage points.  Within the
present noise model, a visibility in the interval
$0.681\lesssim v<0.707$ would be sufficient for the four-dimensional
\QC\ realization to violate its noncontextual bound while the corresponding
Wagner realization would remain below threshold.

The comparison with $C_7$ also shows why the unrestricted contextuality gap
$\Delta(G)$ should not be interpreted directly as a noise threshold.
Although
$\Delta(C_7)>\Delta(C_5)$, the seven-cycle has the larger critical
visibility,
\begin{equation}
  v^*(C_7)\approx0.677
  >
  v^*(C_5)\approx0.585.
\end{equation}
Equation~\eqref{eq:vstar} makes the reason explicit: robustness depends on
the positions of both $\alpha(G)$ and $\eta_d(G)$ relative to the
completely mixed-state value $n/d$, rather than on their difference alone.
The critical visibility therefore provides a more directly operational
measure of robustness for a specified Hilbert-space realization.

The values in Table~\ref{tab:noise} should be interpreted as benchmarks for
depolarizing noise on the prepared state.  A complete experimental error
budget would additionally include imperfections in the projective
measurements, compatibility or exclusivity relations, and other
platform-dependent sources of noise.

\section{Experimental considerations and possible implementations}
\label{sec:ExperimentalFeasibility}

The dimension-dependent results above also provide a concrete route toward
experimental tests.  The KCBS pentagon is attractive because it realizes
contextuality with only five measurement events in a three-dimensional
Hilbert space~\cite{Klyachko2008}.  The \QC\ construction asks a complementary
question: how much stronger can a compact contextuality witness become when
the number of measurement events and the available Hilbert-space dimension
are increased modestly?  Its eight-event structure remains sufficiently
small that the required states and projective measurements can be specified
explicitly.

For the qutrit realization, the experiment requires eight rank-one
projectors
\begin{equation}
  \Pi_i=|v_i\rangle\langle v_i|,
  \qquad i=0,\ldots,7,
\end{equation}
with
\begin{equation}
  \Pi_i\Pi_j=0
  \qquad\text{for every }(i,j)\in E.
\end{equation}
The vectors constructed in Appendix~\ref{app:analytic} therefore specify the
ideal measurement geometry.  For a fixed realization, the state maximizing
the witness is the eigenstate corresponding to the largest eigenvalue of
\begin{equation}
  M=\sum_{i=0}^{7}\Pi_i .
\end{equation}
For the qutrit construction this gives
\begin{equation}
  \sum_{i=0}^{7}\langle\Pi_i\rangle
  =
  \eta_3(\QC)
  =
  1+\sqrt{5},
\end{equation}
with the corresponding optimal state given in
Appendix~\ref{app:analytic}.

The qutrit realization is non-injective: in the explicit construction,
two pairs of graph vertices are represented by identical projectors,
$P_0=P_6$ and $P_2=P_3$.  This is allowed by the orthogonal-representation
formulation because nonadjacent vertices are not required to correspond to
distinct projectors.  Experimentally, these repeated projectors would have
to be implemented consistently across the different measurement contexts.
Thus, the construction specifies the target projective measurements, while
an actual contextuality test must additionally verify the required
compatibility or exclusivity relations and control context-dependent
measurement errors.

Several existing platforms already provide the ingredients needed for such
three-dimensional measurements.  Contextuality has been demonstrated with
single photonic qutrits~\cite{Lapkiewicz2011,Amselem2009}, superconducting
qutrits~\cite{Jerger2016}, trapped-ion systems~\cite{Kirchmair2009}, and
NMR qutrits~\cite{Dogra2016}.  Photonic implementations are particularly
flexible because three-dimensional subspaces can be encoded in spatial,
path, polarization--path, or orbital-angular-momentum modes
~\cite{Molina-Terriza2007}.  These experiments establish that preparation
and sequential or context-resolved measurement of multilevel quantum systems
at the dimensionality relevant here are experimentally accessible.

The four-dimensional realization should likewise be interpreted as a
Hilbert-space requirement rather than as a specific hardware architecture.
It may be implemented in a single four-level system, a four-mode photonic
encoding, or the joint Hilbert space of two qubits. 
In particular, single-photon experiments combining path and polarization
have already implemented four-dimensional contextuality measurements
~\cite{Amselem2012}.

A two-qubit encoding is another possible realization of the same
four-dimensional state space, but entanglement is not required by the
dimensional classification itself.

For \QC, moving from three to four dimensions has a clear operational
consequence.  The qutrit realization gives the exact contextuality margin
\begin{equation}
  \Delta_3(\QC)=\sqrt{5}-2,
\end{equation}
whereas the four-dimensional realization reaches the full Lov\'asz value
at the reported numerical precision and therefore accesses the larger gap
$\Delta(\QC)\approx0.468$.  The depolarizing-noise analysis of
Sec.~\ref{sec:noise} gives a critical visibility
$v^*(\QC,4)\approx0.681$, compared with
$v^*(W,4)\approx0.707$ for the Wagner graph.  Within that specific
state-depolarization model, the \QC\ realization therefore tolerates a
larger depolarizing fraction before its contextuality violation disappears.

These observations make \QC\ experimentally interesting for two distinct
reasons.  Its qutrit realization provides an eight-event contextuality test
whose exact margin remains KCBS-sized, while its four-dimensional
realization accesses the substantially larger contextuality gap revealed by
the graph search.  A dedicated experiment would still need to choose a
physical encoding, implement the required projectors, and quantify
compatibility and measurement imperfections beyond the ideal depolarizing
model considered here.

\section{Conclusion}

We have performed an exhaustive search over all 11{,}117 connected
non-isomorphic simple graphs with eight vertices and identified \QC, an
eight-vertex, ten-edge exclusivity graph formed from four overlapping KCBS
pentagons, as the graph with the largest contextuality gap
\begin{equation}
  \Delta(G)=\vartheta(G)-\alpha(G)
\end{equation}
within this set.  We find
$\Delta(\QC)\approx0.468$, exceeding the Wagner-graph value
$\Delta(W)\approx0.414$ despite using two fewer exclusivity edges.  The
result therefore shows that a larger contextuality gap need not arise from
a denser exclusivity structure; the organization of the constraints is
equally important.

Beyond the exhaustive numerical search, the \QC\ graph also admits a
simple analytical characterization in three dimensions.  We construct an
explicit qutrit orthogonal representation attaining
\begin{equation}
  \eta_3(\QC)=1+\sqrt{5},
\end{equation}
and prove that this value is optimal for any three-dimensional
representation.  The upper bound follows from the dimensional
orthogonality constraints of the graph: in a qutrit realization one pair
of projectors is necessarily identified, leaving an orthonormal basis
together with a KCBS pentagon whose contribution is bounded by
$\vartheta(C_5)=\sqrt{5}$.  Consequently,
\begin{equation}
  \Delta_3(\QC)=\sqrt{5}-2,
\end{equation}
exactly equal to the KCBS contextuality margin.

The dimensional dependence is nevertheless nontrivial.  In four
dimensions our numerical realization reaches the unrestricted Lov\'asz
value at the reported precision, making the substantially larger gap
$\Delta(\QC)\approx0.468$ accessible.  Thus the same eight-event
exclusivity structure connects an analytically solvable qutrit
contextuality witness to a stronger four-dimensional realization.
By contrast, our numerical search finds no Wagner-graph violation in
three dimensions and a violation in four dimensions, consistent with a
higher dimensional requirement for that witness.

This dimension dependence also has an operational consequence under the
state-depolarization model considered here.  The qutrit realization of
\QC\ has the same critical visibility as KCBS,
\begin{equation}
  v^*
  =
  \frac{1}{3\sqrt{5}-5}
  \approx0.585,
\end{equation}
while the four-dimensional realization gives
$v^*(\QC)\approx0.681$, compared with
$v^*(W)\approx0.707$ for the Wagner graph.  These values provide simple
benchmarks for experimental implementations, while a complete experiment
would additionally need to account for measurement imperfections and
compatibility errors beyond the ideal depolarizing model.

More broadly, the eight-vertex optimum points to overlapping odd-cycle
structures as a useful motif for compact contextuality witnesses.  In
\QC, four KCBS pentagons share the same set of exclusivity relations in a
highly organized way, producing a graph that is sparse, qutrit-accessible,
and capable of a substantially larger contextuality gap in four dimensions.
Whether similar overlapping-cycle constructions provide a systematic route
to stronger witnesses at larger graph size is a natural direction for
future work.

The relatively small number of measurement events and the explicit
three- and four-dimensional realizations also make \QC\ a concrete
candidate for experimental investigation in multilevel quantum systems.
More broadly, the result provides a compact setting in which the strength
of contextuality can be studied together with the Hilbert-space resources
required to realize it.  This may be useful for applications in which
contextuality is regarded as a resource for quantum information processing,
communication, or metrology, and for understanding how such resources
develop as measurement scenarios become larger.  Extending the present
search to higher vertex numbers and identifying whether overlapping
odd-cycle structures continue to generate strong witnesses could therefore
provide a systematic route toward scalable contextuality constructions and
their experimental implementation.

\section*{Data Availability}

All graph data derive from the publicly available McKay database
(\texttt{graph8.g6},
\url{https://users.cecs.anu.edu.au/~bdm/data/}).
Complete edge lists are given in Table~\ref{tab:edgelists} (Appendix~\ref{app:tables});
Python scripts and data files supporting this work are publicly available at
\url{https://github.com/ugurtamerphys/quad-c5-contextuality} and archived on
Zenodo~\cite{Tamer2026zenodo}.

\begin{acknowledgments}
We thank Ad\'{a}n Cabello for helpful correspondence and for drawing our
attention to earlier graph-database studies of contextuality and their
experimental realizations.
We also thank Sena Nur Tamer, Julie Zeynep Pagy, and Serta\c{c} F{\i}rat
for useful discussions.
Part of this work was carried out at Bilimler K\"{o}y\"{u}.
The numerical calculations used the open-source packages CVXPY, NetworkX,
SCS, CLARABEL, and mpmath~\cite{mpmath}.
The authors used artificial-intelligence-assisted language-editing tools to improve grammar,
clarity, and readability of parts of the manuscript.
All scientific ideas, calculations, analysis, interpretation of results,
conclusions, and final text were developed, checked, and approved by the
authors.
\end{acknowledgments}

\appendix

\section{Explicit qutrit construction for \QC}
\label{app:analytic}

This appendix provides the explicit three-dimensional realization used in
Sec.~\ref{sec:eta3_exact} to establish the lower bound
$\eta_3(\QC)\geq1+\sqrt{5}$.  The graph itself, its edge set, and its
decomposition into four overlapping KCBS pentagons are given in
Sec.~\ref{sec:qc_structure}; they are not repeated here.

\subsection{Orthogonal representation}

Let
\begin{equation}
  \varphi=\frac{1+\sqrt{5}}{2}
\end{equation}
denote the golden ratio.  Consider the following unit vectors in
$\mathbb{R}^3$:
\begin{equation}
\begin{aligned}
  v_0 &= v_6 = (1,\;0,\;0),\\
  v_1 &= \bigl(0,\;\varphi^{-1/2},\;\varphi^{-1}\bigr),\\
  v_2 &= v_3 = (0,\;0,\;1),\\
  v_4 &= \bigl(-\varphi^{-1},\;-\varphi^{-3/2},
                \;\varphi^{-1}\bigr),\\
  v_5 &= (0,\;1,\;0),\\
  v_7 &= \bigl(-\varphi^{-1},\;\varphi^{-1/2},\;0\bigr).
\end{aligned}
\label{eq:vectors_app}
\end{equation}
Here
$\varphi^{-1/2}=\sqrt{(\sqrt{5}-1)/2}$ and
$\varphi^{-1}=(\sqrt{5}-1)/2$.

The coincident pairs $v_0=v_6$ and $v_2=v_3$ are allowed because the
corresponding vertex pairs are nonadjacent.  Thus the construction is a
valid, although non-injective, orthogonal representation of the exclusivity
graph.

We now verify the two requirements for an orthogonal representation.  The
vectors have unit norm.  For the nontrivial cases, the golden-ratio identity
$\varphi^{-1}+\varphi^{-2}=1$ gives
\begin{align}
  \|v_1\|^2
  &=\varphi^{-1}+\varphi^{-2}=1,\\
  \|v_7\|^2
  &=\varphi^{-2}+\varphi^{-1}=1,
\end{align}
while
\begin{equation}
  \|v_4\|^2
  =
  2\varphi^{-2}+\varphi^{-3}
  =
  1.
\end{equation}
The remaining vectors are basis vectors and therefore have unit norm
immediately.

Substitution into the edge constraints of \QC\ verifies
\begin{equation}
  \langle v_i|v_j\rangle=0
  \qquad
  \text{for every }(i,j)\in E.
\end{equation}
Most of these relations follow directly from the zero components in
Eq.~\eqref{eq:vectors_app}.  The two nontrivial cancellations are
\begin{align}
  \langle v_1|v_4\rangle
  &=
  -\varphi^{-2}+\varphi^{-2}=0,\\
  \langle v_4|v_7\rangle
  &=
  \varphi^{-2}-\varphi^{-2}=0.
\end{align}
Equation~\eqref{eq:vectors_app} therefore defines an explicit qutrit
orthogonal representation of \QC.

\subsection{Witness value and optimal state}

For this representation, define the witness operator
\begin{equation}
  M=\sum_{i=0}^{7}|v_i\rangle\langle v_i|.
\end{equation}
Direct evaluation gives the characteristic polynomial in factorized form,
\begin{equation}
  \det(\lambda I-M)
  =
  \bigl[\lambda-(1+\sqrt{5})\bigr]
  \left[
    \lambda-\frac{7-\sqrt{5}}{2}
  \right]^2 .
  \label{eq:charpoly_app}
\end{equation}
Since
\begin{equation}
  1+\sqrt{5}
  >
  \frac{7-\sqrt{5}}{2},
\end{equation}
the largest eigenvalue is
\begin{equation}
  \lambda_{\max}(M)=1+\sqrt{5}.
\end{equation}
The explicit construction therefore establishes
\begin{equation}
  \eta_3(\QC)\geq1+\sqrt{5}.
\end{equation}
Together with the upper bound proved in
Sec.~\ref{sec:eta3_exact}, this yields the exact result
$\eta_3(\QC)=1+\sqrt{5}$.

A normalized state attaining this value is
\begin{equation}
  |\psi_\star\rangle
  =
  \frac{1}{5^{1/4}}
  \begin{pmatrix}
    -1\\
    \sqrt{\sqrt{5}-2}\\
    1
  \end{pmatrix},
  \label{eq:psi_star}
\end{equation}
which satisfies
\begin{equation}
  M|\psi_\star\rangle
  =
  (1+\sqrt{5})|\psi_\star\rangle.
\end{equation}
Consequently,
\begin{equation}
  \sum_{i=0}^{7}
  |\langle\psi_\star|v_i\rangle|^2
  =
  1+\sqrt{5}.
\end{equation}
Thus the construction supplies both the projective measurements and the
state that attain the exact qutrit value used in the main text.

\section{Numerical diagnostics and reproducibility data}
\label{app:tables}

This appendix collects numerical information supporting the graph search
and the dimension-restricted optimizations.  These data are separated from
the main text because they concern numerical stability and reproducibility
rather than the physical interpretation of the results.

\subsection{Stability of the semidefinite-program calculations}

The unrestricted Lov\'asz values used to rank the eight-vertex graphs were
computed independently with SCS and CLARABEL.  For the leading candidates,
we additionally repeated the high-accuracy calculations to test whether
solver-level numerical variation could affect their ordering.

Table~\ref{tab:intervals} shows the range of values obtained in these
repeated calculations.  The separation between \QC\ and the second-ranked
graph is approximately $2.94\times10^{-2}$, many orders of magnitude larger
than the observed solver-level variations.

\begin{table}[tb]
  \caption{\label{tab:intervals}
  Numerical ranges of $\vartheta(G)$ obtained from five high-accuracy runs
  with each of the two independent SDP solvers, SCS and CLARABEL.
  ``Minimum'' and ``maximum'' denote the smallest and largest numerical
  values returned by these runs; they are not rigorous mathematical
  interval bounds.}
  \begin{ruledtabular}
  \begin{tabular}{lccc}
    Graph & Minimum & Maximum & Half-width \\
    \hline
    \QC\ (Rank 1) & 3.46784373 & 3.46784378 & $2.4\times10^{-8}$ \\
    Rank 2        & 3.43844718 & 3.43844719 & $6.3\times10^{-10}$ \\
    Wagner $W$    & 3.41421356 & 3.41421364 & $3.7\times10^{-8}$ \\
  \end{tabular}
  \end{ruledtabular}
\end{table}

For completeness, Table~\ref{tab:certification} reports the detailed
high-accuracy SDP diagnostics for the 20 leading graphs.  Here
$r_p=|\Tr X-1|$ measures violation of the trace constraint,
$r_{\rm edge}$ is the largest residual of an edge constraint,
$\lambda_{\min}(X)$ monitors positive-semidefinite feasibility,
$r_d$ is the primal--dual objective difference, and
$\lambda_{\min}(S)$ is the smallest eigenvalue of the dual slack matrix.
Values close to zero at the scale shown are consistent with the numerical
precision of the SDP solutions.  The two independent solvers agree to
better than $10^{-7}$ for all 50 candidates recomputed at high accuracy.

\begin{table*}[!t]
  \caption{\label{tab:certification}
  High-accuracy SDP diagnostics for the 20 highest-ranked connected
  eight-vertex graphs.  The values shown are from the CLARABEL
  calculations; the residual quantities are defined in the accompanying
  text.}
  \begin{ruledtabular}
  \scriptsize
  \begin{tabular}{clrccccc}
    Rank & Note &
    $\vartheta_{\rm CL}$ &
    $r_p$ &
    $r_{\rm edge}$ &
    $\lambda_{\min}(X)$ &
    $r_d$ &
    $\lambda_{\min}(S)$ \\
    \hline
     1 & \QC\          & 3.46784 & $+3.7\times10^{-17}$ & $+3.7\times10^{-17}$ & $-5.6\times10^{-9}$ & $-8.5\times10^{-9}$ & $-2.5\times10^{-8}$ \\
     2 &               & 3.43845 & $+1.2\times10^{-19}$ & $+1.2\times10^{-19}$ & $-1.3\times10^{-9}$ & $-1.2\times10^{-9}$ & $-6.8\times10^{-9}$ \\
     3 & Wagner $W$    & 3.41421 & $+2.1\times10^{-19}$ & $+2.1\times10^{-19}$ & $-5.7\times10^{-9}$ & $+2.5\times10^{-9}$ & $-1.2\times10^{-8}$ \\
     4 &               & 3.37228 & $+4.4\times10^{-16}$ & $+3.7\times10^{-17}$ & $-4.8\times10^{-9}$ & $-5.6\times10^{-9}$ & $+4.9\times10^{-9}$ \\
     5 &               & 3.37228 & $+4.2\times10^{-18}$ & $+4.3\times10^{-18}$ & $-2.6\times10^{-9}$ & $-3.3\times10^{-9}$ & $+2.0\times10^{-9}$ \\
     6 &               & 3.37228 & $+8.8\times10^{-15}$ & $+1.8\times10^{-15}$ & $-3.9\times10^{-10}$ & $-4.6\times10^{-10}$ & $-2.6\times10^{-9}$ \\
     7 &               & 2.34315 & $+3.1\times10^{-15}$ & $+9.4\times10^{-16}$ & $-6.6\times10^{-11}$ & $+1.3\times10^{-11}$ & $-1.5\times10^{-10}$ \\
     8 &               & 3.33804 & $+1.4\times10^{-14}$ & $+1.4\times10^{-14}$ & $-6.6\times10^{-10}$ & $-1.6\times10^{-9}$ & $-3.2\times10^{-9}$ \\
     9 &               & 2.33804 & $+2.2\times10^{-16}$ & $+3.4\times10^{-18}$ & $-2.3\times10^{-9}$ & $+3.0\times10^{-10}$ & $-3.8\times10^{-9}$ \\
    10 &               & 2.33333 & $+5.6\times10^{-18}$ & $+5.7\times10^{-18}$ & $-7.7\times10^{-10}$ & $-3.4\times10^{-10}$ & $-2.5\times10^{-9}$ \\
    11 &               & 3.33097 & $+1.9\times10^{-19}$ & $+1.9\times10^{-19}$ & $-1.1\times10^{-9}$ & $-2.0\times10^{-9}$ & $+1.9\times10^{-9}$ \\
    12 &               & 3.32922 & $+4.1\times10^{-18}$ & $+4.1\times10^{-18}$ & $-4.3\times10^{-9}$ & $-7.5\times10^{-9}$ & $+7.8\times10^{-9}$ \\
    13 &               & 3.31767 & $+9.2\times10^{-19}$ & $+9.2\times10^{-19}$ & $-2.2\times10^{-9}$ & $-5.2\times10^{-9}$ & $-1.7\times10^{-8}$ \\
    14 &               & 3.31767 & $+9.4\times10^{-19}$ & $+9.4\times10^{-19}$ & $-1.7\times10^{-9}$ & $-4.9\times10^{-9}$ & $-9.5\times10^{-9}$ \\
    15 &               & 3.31767 & $+2.2\times10^{-16}$ & $+2.7\times10^{-18}$ & $-1.5\times10^{-9}$ & $-2.1\times10^{-9}$ & $-5.6\times10^{-9}$ \\
    16 &               & 3.31767 & $+4.9\times10^{-11}$ & $+4.9\times10^{-11}$ & $-2.8\times10^{-9}$ & $-8.1\times10^{-9}$ & $-2.0\times10^{-8}$ \\
    17 &               & 3.31767 & $+3.1\times10^{-10}$ & $+3.1\times10^{-10}$ & $-1.9\times10^{-9}$ & $-3.4\times10^{-9}$ & $-6.5\times10^{-9}$ \\
    18 &               & 3.31767 & $+3.1\times10^{-11}$ & $+3.1\times10^{-11}$ & $-3.3\times10^{-9}$ & $-4.4\times10^{-9}$ & $-1.7\times10^{-8}$ \\
    19 &               & 3.31767 & $+1.2\times10^{-13}$ & $+2.3\times10^{-15}$ & $-8.1\times10^{-10}$ & $-2.3\times10^{-9}$ & $-2.2\times10^{-9}$ \\
    20 &               & 3.31767 & $+4.1\times10^{-11}$ & $+4.1\times10^{-11}$ & $-8.9\times10^{-10}$ & $-1.7\times10^{-9}$ & $-7.6\times10^{-10}$ \\
  \end{tabular}
  \end{ruledtabular}
\end{table*}

\subsection{Dimension-restricted optimization}

Table~\ref{tab:dim_diagnostics}(a) summarizes the dimension-restricted
results for the six leading eight-vertex graphs. For \QC, the qutrit value is exact by
Proposition~\ref{prop:eta3_exact}.  For the other graphs,
the listed $\eta_3$ values are the best feasible values found by the
nonconvex optimization of Sec.~\ref{sec:dimension_classification}.
Whenever $\eta_3>\alpha$, the numerical construction itself proves that a
qutrit realization is contextual.  Conversely, failure to find
$\eta_3>\alpha$ does not constitute an upper-bound proof.

\begin{table*}[!t]
\caption{\label{tab:dim_diagnostics}
Dimension-restricted results and numerical diagnostics for the leading
eight-vertex graphs.  Panel (a) summarizes the dimensional classification;
panel (b) gives representative fixed-dimension optimization diagnostics.
For nonconvex searches, the reported values are constructive lower bounds
unless an exact result is established separately.}
\begin{ruledtabular}
\begin{minipage}[t]{0.48\textwidth}
\centering
\textbf{(a) Dimension classification}\\[4pt]

\begin{tabular}{crrcl}
Rank & $\Delta$ & $\eta_3$ & $d^*$ & Status \\
\hline
1 (\QC)      & 0.46784 & 3.23607 & 3            & Qutrit confirmed \\
2            & 0.43845 & 3.23607 & 3            & Qutrit confirmed \\
3 (Wagner)   & 0.41421 & 3.00000 & $4^\dagger$  & 4D indicated \\
4            & 0.37228 & 3.23607 & 3            & Qutrit confirmed \\
5            & 0.37228 & 3.00000 & $4^\dagger$  & 4D indicated \\
6            & 0.37228 & 3.33804 & 3            & Qutrit confirmed \\
\end{tabular}

\vspace{3pt}
\raggedright\footnotesize
$^\dagger$Numerically indicated; a rigorous three-dimensional
upper bound is not established here.
\end{minipage}
\hfill
\begin{minipage}[t]{0.48\textwidth}
\centering
\textbf{(b) Optimization diagnostics}\\[4pt]

\begin{tabular}{lrrcl}
Configuration & $d$ & $\lambda_{\max}$ &
Contextual? & Max.\ orth.\ error \\
\hline
\QC, $d=3$    & 3 & 3.236068 & Yes & $<2.3\times10^{-14}$ \\
\QC, $d=4$    & 4 & 3.467844 & Yes & $<3.3\times10^{-14}$ \\
Wagner, $d=3$ & 3 & 3.000000 & No violation & $<6.0\times10^{-15}$ \\
Wagner, $d=4$ & 4 & 3.414214 & Yes & $<5.8\times10^{-14}$ \\
\end{tabular}
\end{minipage}
\end{ruledtabular}
\end{table*}

Table~\ref{tab:dim_diagnostics}(b) gives the corresponding optimization
diagnostics for the configurations most relevant to the dimension comparison. The
orthogonality error is the largest value of
$|\langle v_i|v_j\rangle|$ over all graph edges after the final constrained
refinement.

\subsection{Edge lists of the leading graphs}

For reproducibility, Table~\ref{tab:edgelists} gives the complete edge
sets of the six highest-ranked eight-vertex graphs.  The edge set of \QC\
is also displayed in the main text because it enters the analytical
qutrit argument directly.

\begin{table*}[!t]
  \caption{\label{tab:edgelists}
  Edge lists for the six highest-ranked connected eight-vertex
  contextuality witnesses.  The contextuality gaps $\Delta$ are obtained
  from the high-accuracy SDP calculations.  Daggers denote cases for which $d^*=4$ is numerically indicated, as
explained in Table~\ref{tab:dim_diagnostics}(a).}
  \begin{ruledtabular}
  \scriptsize
  \begin{tabular}{cllll}
    Rank & $\Delta$ & $d^*$ & $|E|$ & Edge list \\
    \hline
    1 & 0.46784 & 3 & 10 &
      $(0,3),(0,5),(1,4),(1,6),(2,5),(2,6),(2,7),(3,6),(3,7),(4,7)$ \\
    2 & 0.43845 & 3 & 11 &
      $(0,3),(0,5),(1,4),(1,5),(1,7),(2,5),(2,6),(2,7),(3,6),(3,7),(4,6)$ \\
    3 (Wagner $W$) & 0.41421 & $4^\dagger$ & 12 &
      $(0,1),(1,2),(2,3),(3,4),(4,5),(5,6),(6,7),(7,0),(0,4),(1,5),(2,6),(3,7)$ \\
    4 & 0.37228 & 3 & 11 &
      $(0,3),(0,5),(0,7),(1,4),(1,6),(2,5),(2,6),(2,7),(3,6),(3,7),(4,7)$ \\
    5 & 0.37228 & $4^\dagger$ & 12 &
      $(0,3),(0,5),(1,4),(1,5),(1,6),(2,5),(2,6),(2,7),(3,6),(3,7),(4,7),(5,7)$ \\
    6 & 0.37228 & 3 & 10 &
      $(0,3),(0,4),(0,7),(1,4),(1,5),(2,5),(2,6),(2,7),(3,6),(4,7)$ \\
  \end{tabular}
  \end{ruledtabular}
\end{table*}

\end{document}